\documentclass[runningheads]{llncs}

\usepackage{color}
\usepackage{listings}
\usepackage{url}
\usepackage{enumitem}
\usepackage{mathtools}
\usepackage{braket}
\usepackage{amsmath}
\usepackage{todonotes}
\usepackage{makecell}
\usepackage{amsfonts}
\usepackage[hidelinks]{hyperref}
\usepackage{mathalfa,amssymb}
\usepackage{tablefootnote}

\usepackage{tabularx,array,booktabs}

\definecolor{graybright}{cmyk}{0, 0, 0, 0.15}
\definecolor{mydarkgreen}{cmyk}{0.85, 0.31, 0.96, 0.2}
\definecolor{myblue}{cmyk}{0.91, 0.67, 0.53, 0.51}
\definecolor{myred}{cmyk}{0.08, 0.86, 0.75, 0.01}
\definecolor{myorange}{cmyk}{0.08, 0.49, 1, 0}
\definecolor{stringgreen}{rgb}{0.25,0.5,0.35} 
\definecolor{mygray}{rgb}{0.5, 0.5, 0.5} 

\usepackage{pifont}

\usepackage{tikz}
\usetikzlibrary{automata, positioning, arrows, fit}
\tikzset{->,  
>=stealth, 
node distance=0.1cm, 
}

\tikzset{every state/.append style={thick, font=\scriptsize, fill=gray!10, rectangle, rounded corners, inner sep=3pt, minimum size=0.4cm}}

\tikzset{hidden/.append style={thick, font=\scriptsize, text=gray!50, draw=gray!20, rectangle, rounded corners, inner sep=3pt, minimum size=0.4cm}}

\tikzset{label_node/.style={font=\scriptsize, align=center, auto}}

\newcommand*{\Scale}[2][4]{\scalebox{#1}{$#2$}}%

\providecommand{\incomp}{
\Scale[0.7]{
  \not\mathrel{
  \smash{
  \vcenter{
    \offinterlineskip 
    \ialign{
       \hfil##\hfil\cr 
       $\sqsubseteq$\cr 
       \noalign{\kern.3ex}
       $\sqsupset$\cr 
    }
  }
  }
  \vphantom{>}
  }
 }
}

\newcommand{\CodeCurly}[1]{\textcolor{mydarkgreen}{#1}}
\newcommand{\CodeRound}[1]{\textcolor{myred}{#1}}

\usepackage{adjustbox}

\usepackage{letltxmacro}
\newcommand*{\SavedLstInline}{}
\LetLtxMacro\SavedLstInline\lstinline
\DeclareRobustCommand*{\lstinline}{%
  \ifmmode
    \let\SavedBGroup\bgroup
    \def\bgroup{%
      \let\bgroup\SavedBGroup
      \hbox\bgroup
    }%
  \fi
  \SavedLstInline
}

\lstdefinestyle{BCSL}{
  literate={(}{{\CodeRound{(}}}1
           {=>}{{$\Rightarrow$}}1
           {+}{{\hspace{0.05cm}$+$\hspace{0.05cm}}}1
           {::}{{\hspace{0.01cm}$::$\hspace{0.01cm}}}1
           {*}{{$\times$}}1
           {~}{{$\sim$}}1
           {)}{{\CodeRound{)}}}1
           {\{}{{\CodeCurly{\{}}}1
           {\}}{{\CodeCurly{\}}}}1, 
  commentstyle=\color{mygray}\ttfamily
}

\lstdefinestyle{EBNF}{
  stringstyle=\color{stringgreen},
  showstringspaces=false
}

\usepackage{geometry}
 \usepackage{subcaption}

\begin{document}

\title{Biochemical Space Language \\in Relation to Multiset Rewriting Systems} 

\titlerunning{The Relation of Biochemical Space Language to Multiset Rewriting Systems} 

\author{Matej Troj\'{a}k, David \v{S}afr\'{a}nek, and Lubo\v{s} Brim}

\institute{Faculty of Informatics, Masaryk University, Brno, Czech Republic}

\authorrunning{Troj\'ak et al.} 

\maketitle

\begin{abstract}
This technical report relates Biochemical Space Language (BCSL)~\cite{TROJAK202091} to Multiset rewriting systems (MRS)~\cite{trojak2021regulated}. For a BCSL model, the semantics are defined in terms of transition systems, while for an MRS, they are defined in terms of a set of runs. In this report, we relate BCSL to MRS by first showing how the transition system is related to a set of runs and consequently showing how for every BCSL model, an MRS can be constructed such that both represent the same set of runs. The motivation of this step is to establish BCSL in the context of a more general rewriting system and benefit from properties shown for them. Finally, we show that \emph{regulations} defined for MRS can be consequently used in the BCSL model.
\end{abstract}

\section{Multiset rewriting systems}

This section recalls some definitions and known results about multisets and rewriting systems over them. Intuitively, a multiset is a set of elements with allowed repetitions. A~multiset rewriting \emph{rule} describes how a particular multiset is transformed into another one. A multiset rewriting system consists of a set of rewriting rules, defining how the system can evolve, and an initial multiset, representing the starting point for the rewriting.

\begin{definition}{Multiset}

\noindent Let $\mathcal{S}$ be a finite set of \emph{elements}. A \emph{multiset} over $\mathcal{S}$ is a total function $ \mathtt{M} : \mathcal{S} \rightarrow \mathbb{N}$ (where $\mathbb{N}$ is the set of natural numbers including 0). For each  $\mathit{a} \in \mathcal{S}$ the \emph{multiplicity} (the number of occurrences) of $\mathit{a}$ is the number $\mathtt{M}(\mathit{a})$. 
\end{definition}

Operations and relations over multisets are defined in a standard way, taking into account the repetition of elements.

\begin{definition}{Operations and relations over multisets}

\begin{center}
\begin{tabular}{l l}
-- Union & $\mathtt{M}_1 \cup \mathtt{M}_2:~\forall a \in \mathcal{S}.~ (\mathtt{M}_1 \cup \mathtt{M}_2)(a) = \mathtt{M}_1(a) + \mathtt{M}_2(a)$\\
-- Difference & $\mathtt{M}_1 \smallsetminus \mathtt{M}_2$:\\
  & $\forall a \in \mathcal{S}.~ (\mathtt{M}_1 \smallsetminus \mathtt{M}_2)(a) = $
    $\begin{cases}
    \mathtt{M}_1(a) - \mathtt{M}_2(a) & \mathtt{if}~ \mathtt{M}_2(a) \leq \mathtt{M}_1(a) \\
    0 & \mathtt{otherwise}
  \end{cases} $\\
-- Intersection & $\mathtt{M}_1 \cap \mathtt{M}_2: \forall a \in \mathcal{S}.~ (\mathtt{M}_1 \cap \mathtt{M}_2)(a) = \mathtt{min}\{\mathtt{M}_1(a), \mathtt{M}_2(a)\} $ \\
-- Submultiset & $ \mathtt{M}_1 \subseteq \mathtt{M}_2: \forall a \in \mathcal{S}.~ \mathtt{M}_1(a) \leq \mathtt{M}_2(a) $\\
-- Equality & $\mathtt{M}_1 = \mathtt{M}_2: \forall a \in \mathcal{S}.~ \mathtt{M}_1(a) = \mathtt{M}_2(a) $\\
-- Occurrence & $a \in \mathtt{M}: \exists a \in \mathcal{S}.~ \mathtt{M}(a) \geq 1 $\\
\end{tabular}
\end{center}
\end{definition}

\begin{definition}{Multiset rewriting rule}

\noindent A \emph{multiset rewriting rule} over $\mathcal{S}$ is a pair $\mu = (^{\bullet}\mu, \mu^{\bullet})$ of multisets over $\mathcal{S}$, usually written as $\mu: \hspace{0.01cm} ^{\bullet}\mu \to \mu^{\bullet}$ for convenience.
\end{definition}

The rule rewrites elements specified in the left-hand multiset $^{\bullet}\mu$ to elements specified in the right-hand multiset $\mu^{\bullet}$.

\begin{definition}{Multiset rewriting system}

\noindent A \emph{multiset rewriting system (MRS)} over $\mathcal{S}$ is a pair $\mathcal{M} = (\mathcal{X}, \mathtt{M_0})$, where $\mathcal{X}$ is a finite set of multiset rewrite rules and $\mathtt{M_0}$ is the initial multiset \emph{(state)}, both over $\mathcal{S}$.
\end{definition}

We denote by $\mathbb{MRS}$ the class of multiset rewriting systems. 

\begin{definition}{Enabled rule}

\noindent Let $\mathtt{M}$ be a multiset and $\mathcal{M} = (\mathcal{X}, \mathtt{M_0})$ an MRS, both over $\mathcal{S}$.
A rule $\mu \in \mathcal{X}$ is \emph{enabled} at $\mathtt{M}$ if ~$^{\bullet}\mu \subseteq \mathtt{M}$. 
\end{definition}

\begin{definition}{Rule application}

\noindent The \emph{application} of an enabled rule $\mu \in \mathcal{X}$ to $\mathtt{M}$, written $\mathtt{M} \rightarrow_{\mu} \mathtt{M}'$, creates a multiset $\mathtt{M}' = (\mathtt{M} \smallsetminus \hspace{0.01cm} ^{\bullet}\mu) \cup \mu^{\bullet}$.
\end{definition}

\begin{definition}{Run}

\noindent A \emph{run} $\pi$ of $\mathcal{M}$ is an infinite sequence of multisets $\pi = \mathtt{M_0} \mathtt{M_1} \mathtt{M_2} \ldots$ such that for any \emph{step} $\mathtt{i > 0}$ holds that $\mathtt{M_{i-1}} \rightarrow_{\mu} \mathtt{M_i}$ for some $\mu \in \mathcal{X}$. We denote by $\pi[\mathtt{i}]$ the multiset created in step $\mathtt{i}$.
\end{definition}

\begin{definition}{Run label}

\noindent A \emph{run label} $\overrightarrow{\pi}$ of a run $\pi$ is an infinite sequence of rules $\overrightarrow{\pi} = \mu_\mathtt{1} \mu_\mathtt{2} \mu_\mathtt{3} \ldots$ such that for any \emph{step} $\mathtt{i > 0}$ holds that $\pi[\mathtt{i-1}] \rightarrow_{\mu_\mathtt{i}} \pi[\mathtt{i}]$. We denote by $\overrightarrow{\pi}[\mathtt{i}]$ the rule applied in step \texttt{i}.
\end{definition}

\begin{definition}{Semantics}

\noindent The \emph{semantics} of system $\mathcal{M}$ is an (infinite) set $\mathfrak{L}(\mathcal{M})$ of all possible runs such that $\forall \pi \in \mathfrak{L}(\mathcal{M}).~ \pi[\mathtt{0}] = \mathtt{M_0}$ (i.e. runs start in the initial multiset).
\end{definition}

To ensure the infiniteness of runs, we implicitly assume the presence of a special \emph{empty} rule $\varepsilon = (\emptyset, \emptyset)$. We require that this rule can be applied \emph{only} when no other rule of the system is enabled. It also ensures that the set of rules $\mathcal{X}$ is always non-empty.

\section{BioChemical Space Language}\label{BCSL-definition}

In this section we provide declarative definition of BioChemical Space Language. A constructive (or \emph{imperative}) version of the definition is available in~\cite{TROJAK202091}.

Let $\mathcal{V}_\delta, \mathcal{V}_a, \mathcal{V}_s, \mathcal{V}_c$ be mutually exclusive finite sets of names of features, names of atomic and structure components, and compartments, respectively. 

\begin{table}
\begin{center}
\begin{tabular}{c}
\begin{lstlisting}[style=EBNF, mathescape]
multiset:    $\emptyset$ | agent | multiset "+" multiset
agent:       chain "::" COMPARTMENT
chain:       component | chain "." component
component:   atomic | structure
structure:   NAME "(" composition ")"
composition: $\emptyset$ | atomic | composition "," atomic
atomic:      NAME "{" FEATURE "}"
\end{lstlisting}
\end{tabular}
\end{center}

\caption{A context-free grammar of core BCSL terms in EBNF~\cite{scowen1993generic} notation.}\label{bcsl_syntax}
\end{table}

In~\autoref{bcsl_syntax} we provide a fragment of complete syntax\protect\footnote{\url{https://github.com/sybila/eBCSgen/wiki/Model-syntax\#complete-syntax}} of BCSL, capturing  agents and multisets, where the terminal (in capitals) $\mathtt{FEATURE} \in \mathcal{V}_\delta$ is from given set of feature names, $\mathtt{NAME} \in \mathcal{V}_a$ (resp. $\mathtt{NAME} \in \mathcal{V}_s$) is from given set of atomic (resp. structure) component names, and $\mathtt{COMPARTMENT} \in \mathcal{V}_c$ is is from given set of compartments. We restrict ourselves only to finite expressions and require that an atomic name occurs at most once in a composition. On top of this syntax, several syntactic extensions~\cite{TROJAK202091} are build providing more convenient and succinct notation.

For simplicity, we denote by $\mathsf{M}$ a multiset and by $\mathbb{M}$ the set of all multisets. We assume the \emph{structural congruence}~$\equiv$ to be the least congruence on terms from~\autoref{terms_congru} satisfying respective axioms. That is, two multisets (or any terms) are \emph{equal} if they are structurally congruent.

\begin{table}
\begin{center}
{\small
\begin{tabular}{|@{\hspace{2em}} l @{\hspace{2em}}|@{\hspace{2em}} l @{\hspace{2em}}|}
\hline
\textbf{Term} & \textbf{Satisfying axioms} \\
\hline
\texttt{multisets} & $\mathsf{M}_1 + \mathsf{M}_2 \equiv \mathsf{M}_2 + \mathsf{M}_1$ \\
 & $\mathsf{M} + \emptyset \equiv \mathsf{M}$\\
 \hline
\texttt{chains} & $\mathtt{chain} . \mathtt{component} \equiv \mathtt{component} . \mathtt{chain} $ \\
\hline
 \texttt{compositions} & $\mathtt{composition} , \mathtt{atomic} \equiv \mathtt{atomic}, \mathtt{composition}$ \\
 & $\emptyset, \mathtt{composition} \equiv \mathtt{composition}$\\
\hline
\end{tabular}
}
\end{center}
\caption{Table defining axioms of structural congruence $\equiv$ for particular terms from BCSL grammar.}\label{terms_congru}
\end{table}

The structural congruence $\equiv$ allows us to formally define the algebraic multiset operations $\in, \subseteq, \subset, \cup, \cap$ and $\setminus$ on BCSL multisets. For example, $\mathtt{agent} \in \mathsf{M}$ corresponds to $\exists \mathsf{M}' \in \mathbb{M} . \mathsf{M} \equiv \mathtt{agent} + \mathsf{M}'$ and $\mathsf{M} \subseteq \mathsf{M}'$ corresponds to $\exists \mathsf{M}'' \in \mathbb{M} . \mathsf{M}' \equiv \mathsf{M} + \mathsf{M}''$. Moreover, by $\mathsf{M}(\mathtt{agent})$ we denote the number of occurrences of agent $\mathtt{agent}$ in the multiset $\mathsf{M}$.

\begin{definition}{Signature}\label{signature}

\noindent\emph{Atomic signature} $\sigma_a: \mathcal{V}_a \rightarrow 2^{\mathcal{V}_\delta}$ is a function from an atomic name to a non-empty set of feature names. Set of possible atomic signatures is denoted as $\Sigma_a$. 
\emph{Structure signature} $\sigma_s: \mathcal{V}_s \rightarrow 2^{\mathcal{V}_a}$ is a function from a structure name to a set of atomic names. Set of possible structure signatures is denoted as $\Sigma_s$. 
\end{definition}

\begin{definition}{Pattern}\label{pattern}

\noindent Let $\mathbb{V}_\delta = \mathcal{V}_\delta \cup \{\varepsilon\}$ be a set of feature names extended by a special symbol $\varepsilon$. \emph{Pattern} $\mathsf{P}$ is defined according to the same grammar as $\mathtt{multiset}$ but with $\mathtt{FEATURE} \in \mathbb{V}_\delta$. We denote by $\mathbb{P}$ the set of all patterns. 
\end{definition}

The two patterns are equal if they are structurally equal (the congruence relation defined on multisets does \emph{not} apply). Finally, a pattern is \emph{well-formed} if the atomics are alphanumerically sorted in compositions with respect to their names. From now on, we assume only well-formed patterns.

\begin{remark}\label{pattern_vs_multiset}
In the following text, there is often a situation when a pattern $\mathsf{P}$  is compared to a multiset $\mathsf{M}$. In such a case, we treat the pattern as a multiset too (i.e. they are equal if they are structurally congruent according to~\autoref{terms_congru}). Moreover, it holds that $\varepsilon \neq \delta$ for any $\delta \in \mathcal{V}_\delta$.
\end{remark}

\begin{definition}{Instantiation}\label{instantiation}

\noindent An \emph{instantiation function} $\mathcal{I}: \mathbb{P} \rightarrow \mathbb{P}$ assigns to every $\mathtt{atomic}$ $\mathsf{A}$ in $\mathbb{P}$ with feature $\varepsilon$ a feature $\delta \in \sigma_a(\mathsf{A})$. By $\Gamma(\mathsf{P})$ we denote a finite set of all possible \emph{instantiations} $\mathcal{I}(\mathsf{P})$ of pattern $\mathsf{P}$.
\end{definition}

We define \emph{deatomisation} of pattern $\mathsf{P}$, written $d(\mathsf{P})$, as a sequence of \texttt{atomic}s preserving the order of their occurence in the pattern. Note that this applies to atomics in both chains and compositions.

\begin{definition}{Consistent instantiations}\label{consistent_instantiations}

\noindent Let us have two finite patterns $\mathsf{P}, \mathsf{P}'$ with deatomisations $d(\mathsf{P}) = \mathsf{A}_1, \mathsf{A}_2, \ldots, \mathsf{A}_n$ and $d(\mathsf{P}') = \mathsf{A}'_1, \mathsf{A}'_2, \ldots, \mathsf{A}'_m$. Next, let us have two instantiations $\mathcal{I}(\mathsf{P}) \in \Gamma(\mathsf{P})$ and $\mathcal{I}(\mathsf{P}') \in \Gamma(\mathsf{P}')$ with their deatomisations $d(\mathcal{I}(\mathsf{P})) = \mathcal{I}(\mathsf{A}_1), \mathcal{I}(\mathsf{A}_2), \ldots, \mathcal{I}(\mathsf{A}_n)$ and $d(\mathcal{I}(\mathsf{P}')) = \mathcal{I}(\mathsf{A}'_1), \mathcal{I}(\mathsf{A}'_2), \ldots, \mathcal{I}(\mathsf{A}'_m)$. We say the instantiations $\mathcal{I}(\mathsf{P}), \mathcal{I}(\mathsf{P}')$ are \emph{consistent}, written
$\mathcal{I}(\mathsf{P})\hspace*{0.2mm}\Delta\hspace*{0.6mm}\mathcal{I}(\mathsf{P}')$, if $\forall i < \mathsf{min}(m,n)$ holds that $\mathsf{A}_i = \mathsf{A}'_i \Rightarrow \mathcal{I}(\mathsf{A}_i) = \mathcal{I}(\mathsf{A}'_i)$.
\end{definition}

Consistency of two instantiations ensures that the same features are assigned in the same positions.

\begin{definition}{Pattern expansion}\label{pattern_expansion}

\noindent\emph{Pattern expansion} is a function $\langle \_ \rangle :
\mathbb{P} \rightarrow \mathbb{P}$ which
extends a given pattern $\mathsf{P}$ to a pattern $\langle \mathsf{P}
\rangle$ such that every occurrence of a composition of a $\mathtt{structure}$ is extended by $\mathtt{atomic}s$ whose names are not yet present in the composition and are defined in the given signature $\sigma_s(\mathtt{structure})$. 
These newly added $\mathtt{atomic}s$ have assigned feature $\varepsilon$ and are inserted to the composition in such way that they preserve the alphanumerical order.
\end{definition}

\begin{definition}{BCSL rule}\label{bcsl_rule}

\noindent A BCSL \emph{rule} $\mathsf{R}$ is a pair $(\mathsf{P}_l,
\mathsf{P}_r) \in \mathbb{P} \times \mathbb{P}$, usually written as $\mathsf{P}_l \rightarrow \mathsf{P}_r$. 
\end{definition}

The rule describes a structural change of a multiset defined by the difference between left-hand and right-hand patterns.

\begin{definition}{BCSL model}\label{bcsl_model}

\noindent A BCSL \emph{model} $\mathcal{B}$ is a tuple $(\mathcal{R}, \sigma_s, \sigma_a, \mathsf{M_0})$ such that $\mathcal{R}$ is a finite set of rewrite rules, $\sigma_s \in \Sigma_s$ is a structure signature, $\sigma_a \in \Sigma_a$ is an atomic signature, and $\mathsf{M_0} \in \mathbb{M}$ is an initial multiset.
\end{definition}

\begin{definition}{BCSL rewriting}\label{bcsl_rewriting}

\noindent Let $\mathcal{B} = (\mathcal{R}, \sigma_s, \sigma_a, \mathsf{M_0})$ be a BCSL model. The \emph{rewriting} of the multisets is given by labelled transition relation $\mathsf{M}_1 \xrightarrow[]{\mathsf{R}} \mathsf{M}_2$ with $\mathsf{M}_1, \mathsf{M}_2 \in \mathbb{M}$ and $\mathsf{R}: \mathsf{P}_l \rightarrow \mathsf{P}_r$ satisfying the following inference rule:

\begin{center}
\begin{tabular}{c}
$\exists~ \mathcal{I}\langle\mathsf{P}_l\rangle \in \Gamma\langle\mathsf{P}_l\rangle.~ \mathcal{I}\langle\mathsf{P}_l\rangle = \mathsf{M}_l$\\
$\exists~ \mathcal{I}\langle\mathsf{P}_r\rangle \in \Gamma\langle\mathsf{P}_r\rangle.~ \mathcal{I}\langle\mathsf{P}_r\rangle = \mathsf{M}_r$\\
$\mathcal{I}\langle\mathsf{P}_l\rangle\hspace*{0.2mm}\Delta\hspace*{0.6mm}\mathcal{I}\langle\mathsf{P}_r\rangle$\\
\hline
$\mathsf{M} + \mathsf{M}_l \xrightarrow[]{\mathsf{R}} \mathsf{M} + \mathsf{M}_r$\\
\end{tabular}
\end{center}
\end{definition}

The rewriting of the multisets gives semantics to the model. Intuitively, for pattern $\mathsf{P}_l$, the corresponding agents from the state are found and consequently replaced according to pattern $\mathsf{P}_r$. Applying such an operation transitively, starting in the initial multiset, yields a labelled transition system (\autoref{unregulated_lts}). An example of an LTS is available in~\autoref{running_lts}, corresponding to an example model from~\autoref{running_model}.

\begin{definition}{Labelled transition system $\mathtt{LTS}$}
\label{unregulated_lts}

\noindent Labelled transition system $\mathtt{LTS}(\mathcal{B}) = (S, T, L)$ of a BCSL model $\mathcal{B}$ is obtained by transitive rewriting of the initial state, where $S$ is a set of states (a state is a multiset of agents), $T$ is a set of transitions (a transition corresponds to the application of a rule), and $L$ is labelling function assigning to each transition an identifier of applied rule. 
\end{definition}

\begin{figure}[!h]

\begin{center}
\begin{tabular}{c}
\begin{lstlisting}[style=BCSL]
#! rules
r1_S ~ P(S{i})::cell => P(S{a})::cell
r1_T ~ P(T{i})::cell => P(T{a})::cell
r2   ~ P()::cell     => P()::out

#! inits
1 P(S{i},T{i})::cell
\end{lstlisting}
\end{tabular}
\end{center}

\caption{\textbf{An example of BCSL model.} A single agent \lstinline[style=BCSL]!P! can be modified on its two active sites, \lstinline[style=BCSL]!S! and \lstinline[style=BCSL]!T!. Both sites can be independently activated in respective rules. Additionally, the agent can be transported to another compartment outside the \lstinline[style=BCSL]!cell!. All rules are labelled using \lstinline[style=BCSL]!label ~! prefix. Initially, there is a single \lstinline[style=BCSL]!P! agent present with both sites inactivated. Please note the signature functions do not need to be explicitly defined as they can be gathered automatically from the model (in this case, $\Sigma_a (\mathtt{M}) = \Sigma_a (\mathtt{T}) = \{ \mathtt{i, a} \}$ and $\Sigma_s (\mathtt{P}) = \{ \mathtt{S, T} \}$).}\label{running_model}

\end{figure}

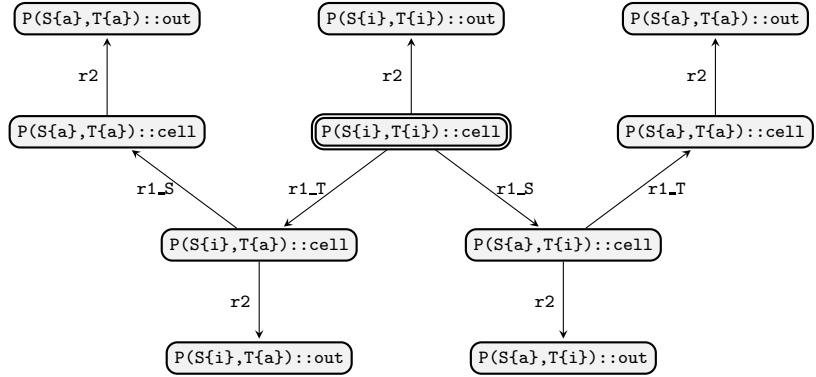
\begin{figure}

\begin{center}
\begin{tikzpicture}[shorten >=1pt, font=\scriptsize]
  \node[state,accepting] at (0,0) (q_0) {\lstinline[style=EBNF]!P(S{i},T{i})::cell!};
  \node[state] at (0,1.5) (q_4) {\lstinline[style=EBNF]!P(S{i},T{i})::out!};
  \node[state] at (-2,-1.5) (q_1) {\lstinline[style=EBNF]!P(S{i},T{a})::cell!};
  \node[state] at (2,-1.5) (q_2) {\lstinline[style=EBNF]!P(S{a},T{i})::cell!};
  \node[state] at (-4,0) (q_3_1) {\lstinline[style=EBNF]!P(S{a},T{a})::cell!};
  \node[state] at (4,0) (q_3_2) {\lstinline[style=EBNF]!P(S{a},T{a})::cell!};
  \node[state] at (-2,-3) (q_5) {\lstinline[style=EBNF]!P(S{i},T{a})::out!};
  \node[state] at (2,-3) (q_6) {\lstinline[style=EBNF]!P(S{a},T{i})::out!};
  \node[state] at (-4,1.5) (q_7_1) {\lstinline[style=EBNF]!P(S{a},T{a})::out!};
  \node[state] at (4,1.5) (q_7_2) {\lstinline[style=EBNF]!P(S{a},T{a})::out!};

  \path[->] (q_0) edge node [left] {$\mathtt{r2}$} (q_4)
            (q_0) edge node [left] {$\mathtt{r1\_T}$} (q_1)
            (q_0) edge node [right] {$\mathtt{r1\_S}$} (q_2)
            (q_1) edge node [left] {$\mathtt{r2}$} (q_5)
            (q_1) edge node [left] {$\mathtt{r1\_S}$} (q_3_1)
            (q_2) edge node [left] {$\mathtt{r2}$} (q_6)
            (q_2) edge node [right] {$\mathtt{r1\_T}$} (q_3_2)
            (q_3_1) edge node [left] {$\mathtt{r2}$} (q_7_1)
            (q_3_2) edge node [left] {$\mathtt{r2}$} (q_7_2);
\end{tikzpicture}
\end{center}

\caption{\textbf{Transition system of the model from~\autoref{running_model} in a tree-like representation.} The double circled state is the initial state, with contents defined in \texttt{inits} part.}\label{running_lts}
\end{figure}

\section{Systems comparison}

This section shows how an MRS can be constructed for any BCSL model and that such an MRS exhibits equivalent behaviour to the original BCSL model.

\subsection{MRS construction}
\label{mrs_construction}

In this section, we show how an MRS $\mathcal{M} = (\mathcal{X}, \mathtt{M_0})$ can be constructed from a BCSL model $\mathcal{B} = (\mathcal{R}, \Sigma_a, \Sigma_s, \mathtt{S_0})$. This approach is based on grounding agents and rules (supplement the missing context -- \autoref{grounding_function}). In particular, we need to do two steps -- construct the support set of elements (\autoref{set_of_elements}) by grounding all possible agents, and then construct the set of multiset rewriting rules by grounding each BCSL rule, creating its possible instantiations in terms of multisets (\autoref{set_of_rules}).

The abstraction provided by BCSL rules allowing to express patterns needs to be grounded in concrete multisets. Informally, this is accomplished by supplementing the context information from the signature functions to the patterns, obtaining particular realisations of patterns. In~\autoref{grounding_function}, there is formal definition of \emph{grounding function} $\Theta$, which uses an instantiation of patterns (\autoref{instantiation}).

\begin{definition}{Grounding function $\Theta$}
\label{grounding_function}

\noindent We define grounding function for a $\mathtt{pattern}$ $\mathsf{P}$ as a set of all its possible instantiated multisets $\Theta(\mathsf{P}) = \{ \mathcal{I}\langle\mathsf{P}\rangle ~|~ \mathcal{I}\langle\mathsf{P}\rangle \in \Gamma\langle\mathsf{P}\rangle \}$. Applied to a $\mathtt{rule}$ $\mathsf{r}$, we obtain a set of all possible \emph{reactions} using consistent instantiations $\Theta(\mathsf{r}) \equiv \Theta(\mathtt{lhs \Rightarrow rhs}) = \{~ \mathcal{I}\langle\mathtt{lhs}\rangle \Rightarrow \mathcal{I}\langle\mathtt{rhs}\rangle ~|~ \mathcal{I}\langle\mathtt{lhs}\rangle \in \Gamma\langle\mathtt{lhs}\rangle \wedge \mathcal{I}\langle\mathtt{rhs}\rangle \in \Gamma\langle\mathtt{rhs}\rangle \wedge \mathcal{I}\langle\mathtt{lhs}\rangle\hspace*{0.2mm}\Delta\hspace*{0.6mm}\mathcal{I}\langle\mathtt{rhs}\rangle \}$, where $\mathcal{I}\langle\mathtt{lhs}\rangle$ and $\mathcal{I}\langle\mathtt{rhs}\rangle$ are treated as multisets (\autoref{pattern_vs_multiset}).
\end{definition}

In~\autoref{set_of_elements}, we show how to create a set of all possible unique $\mathtt{agent}$s present in the model, which can be considered as the set of elements. It is constructed from initial state $\mathtt{M_0}$ and a set of rules $\mathcal{R}$ with the information provided in signature functions. We assume that the initial state $\mathtt{M_0}$ contains agents which are already grounded. 

\begin{definition}{Set of elements}
\label{set_of_elements}

\noindent Let $\mathcal{S}_\mathtt{0}$ be a set of unique elements from initial state $\mathtt{M_0}$ and $\mathcal{S}_\mathcal{R}$ be a set of all possible grounded $\mathtt{agent}$s present in the rules $\mathcal{R}$ defined as $\mathcal{S}_\mathcal{R} = \{ \mathsf{A} \in \Theta(\mathsf{A'}) ~|~ \mathsf{A'} \in \mathcal{A}(\mathtt{r}) \wedge \mathtt{r} \in \mathcal{R} \}$, where $\mathcal{A}(\mathtt{r}) = \mathtt{lhs}~\cup~\mathtt{rhs}$ is a set of all $\mathtt{agent}$s used in rule $\mathtt{r} = \mathtt{lhs \Rightarrow rhs}$. Then, the set of all possible unique $\mathtt{agent}$s present in the model is $\mathcal{S} = \mathcal{S}_\mathtt{0} \cup \mathcal{S}_\mathcal{R}$.
\end{definition}

We show how to construct a set of MRS rewriting rules from a BCSL rule in~\autoref{set_of_rules}. The approach is straightforward since the grounding function $\Theta$ creates the set of all possible grounded rules (reactions). Then, we need to create a pair of multisets from both sides of each grounded rule. The obtained pair of multisets can be directly considered as multiset rewriting rule over support set $\mathcal{S}$, because all the possible agents are already present in the set $\mathcal{S}$ (follows from its construction). We call such rule \emph{MRS instantiation} of the BCSL rule (\autoref{mrs_instantiation}).

\begin{definition}{MRS instantiation}
\label{mrs_instantiation}

\noindent Let $\mathtt{r} = \mathtt{lhs \Rightarrow rhs}$ be a BCSL rule. We define \emph{MRS instantiation} $\mu(\mathtt{r})$ of rule $\mathtt{r}$ as a multiset rewriting rule $\mu(\mathtt{r}) = (\mathtt{L}, \mathtt{R})$ where $\mathtt{L \Rightarrow R}~ \in \Theta(\mathtt{lhs \Rightarrow rhs})$.
\end{definition}

\begin{definition}{Set of rules}
\label{set_of_rules}

\noindent Let $R$ be a set of BCSL rules. The corresponding set of MRS rules $\mathcal{X}$ is defined as a set of all possible MRS instantiations $\mathcal{X} = \{\mu(\mathtt{r}) ~|~ \mathtt{r} \in R \} $.
\end{definition}

We obtain the MRS $\mathcal{M} = (\mathcal{X}, \mathtt{M_0})$ over the set of elements $\mathcal{S}$ (\autoref{set_of_elements}) by taking constructed set of multiset rewriting rules $\mathcal{X}$ (MRS instantiations) as shown in~\autoref{set_of_rules} and the initial state $\mathtt{M_0} $. In~\autoref{running_mrs} there is an example of the MRS constructed from BCSL model (\autoref{running_model}) using this approach.

\begin{figure}[!h]

\centering

\begin{subfigure}[t]{\textwidth}

\centering

$\mathcal{S} = \Set{
  \begin{array}{l}
    \mathtt{P(S\{i\},T\{i\})::cell},~ \mathtt{P(S\{a\},T\{i\})::cell},~
    \mathtt{P(S\{i\},T\{a\})::cell},\\ \mathtt{P(S\{a\},T\{a\})::cell},~
    \mathtt{P(S\{i\},T\{i\})::out},~~~ \mathtt{P(S\{a\},T\{i\})::out},\\
    \mathtt{P(S\{i\},T\{a\})::out},~~~ \mathtt{P(S\{a\},T\{a\})::out}
  \end{array} 
}$

\caption{Set of all unique objects.}

\end{subfigure}%

\begin{subfigure}[t]{\textwidth}

\centering
 \[
 \mathcal{M} = 
    \left\{ 
      \begin{array}{l}
        \mathtt{M_0 = \{\mathtt{P(S\{i\},T\{i\})::cell}\}} \\
        \mathcal{X} = \left\{ 
          \begin{array}{l}
          \mu_\mathtt{r1\_S}: \{ \mathtt{P(S\{i\},T\{i\})::cell} \} \to \{ \mathtt{P(S\{a\},T\{i\})::cell} \}, \\
          \mu_\mathtt{r1\_S}: \{ \mathtt{P(S\{i\},T\{a\})::cell} \} \to \{ \mathtt{P(S\{a\},T\{a\})::cell} \}, \\
          \mu_\mathtt{r1\_T}: \{ \mathtt{P(S\{i\},T\{i\})::cell} \} \to \{ \mathtt{P(S\{i\},T\{a\})::cell} \}, \\
          \mu_\mathtt{r1\_T}: \{ \mathtt{P(S\{a\},T\{i\})::cell} \} \to \{ \mathtt{P(S\{a\},T\{a\})::cell} \}, \\
          \mu_\mathtt{r2}:~~~ \{ \mathtt{P(S\{i\},T\{i\})::cell} \} \to \{ \mathtt{P(S\{i\},T\{i\})::out} \}, \\
          \mu_\mathtt{r2}:~~~ \{ \mathtt{P(S\{a\},T\{i\})::cell} \} \to \{ \mathtt{P(S\{a\},T\{i\})::out} \}, \\
          \mu_\mathtt{r2}:~~~ \{ \mathtt{P(S\{i\},T\{a\})::cell} \} \to \{ \mathtt{P(S\{i\},T\{a\})::out} \}, \\
          \mu_\mathtt{r2}:~~~ \{ \mathtt{P(S\{a\},T\{a\})::cell} \} \to \{ \mathtt{P(S\{a\},T\{a\})::out} \}, \\
          \end{array} 
          \right\}
        \end{array}
      \right\}
\]

\caption{Instantiated rules.}

\end{subfigure}%

\caption{\textbf{MRS representation of a BCSL model from~\autoref{running_model}}. All the objects in set $\mathcal{S}$ are unique strings representing possible forms of original BCSL agents. These are used in multiset rewriting rules and the initial multiset. For convenience, to allow identification of source rule, we label each constructed multiset rewrite rule by $\mu_{\mathtt{r}}$ where \texttt{r} is the label of source BCSL rule.}\label{running_mrs}

\end{figure}

\subsection{BCSL vs. MRS relationship}

In~\autoref{mrs_construction}, we provided an approach to constructing an MRS from any BCSL model. In this section, we show that the behaviour of such a constructed MRS is equivalent to the behaviour of the original BCSL model. This is shown in~\autoref{mrs_is_bcsl} by considering that the type of states in both systems is the same (\autoref{states_equality}), both BCSL rule and its MRS instantiation can always be applied to the same state (\autoref{both_enabled}), and they can always be rewritten to the same states (\autoref{both_apply}).

The semantics of MRS are given in terms of a set of infinite runs, while the semantics of BCSL are given in terms of LTS. First, we need to relate these two constructs. We define how a set of runs corresponds to an LTS. To ensure that the $\mathtt{LTS}$ represents only infinite runs, we extend it to $\mathtt{LTS}_\varepsilon$ such that we add self-loops on states with no successors labelled by an empty rule $\varepsilon$.

\begin{definition}{Run in LTS}
\label{lts_vs_runs}

\noindent Let $\mathtt{LTS}_\varepsilon = (S, T, L)$ be a labelled transition system. $\mathtt{LTS}_\varepsilon$ \emph{generates a set of infinite runs} $\mathfrak{L}(\mathtt{LTS}_\varepsilon)$ such that the infinite run $\pi = s_0, s_1, \ldots, s_n, \ldots$ belongs to $\mathfrak{L}(\mathtt{LTS}_\varepsilon)$ if (i) $s_0 \in S$ and (ii) for all $i \geq 1: (s_{i-1}, s_i) \in T$. Moreover, such a run $\pi$ has a run label $\overrightarrow{\pi} = l_1, l_2, \ldots, l_n, \ldots$ such that for all $i \geq 1: L((s_{i-1}, s_i)) = l_i$.
\end{definition}

\begin{remark}\label{states_equality}
Multisets in constructed MRS use as elements grounded BCSL \texttt{agent}s, and therefore the type of MRS multisets and BCSL multisets is the same and they can be freely interchanged and checked for equality.
\end{remark}

From the construction of $\mathcal{X}$ (\autoref{set_of_rules}) follows that for any rule $\mathtt{lhs \Rightarrow rhs} \in R$, the function $\Theta$ creates grounded rules, which represent all possible instantiations. Then, for any instantiation $\mathtt{L \Rightarrow R}$, \texttt{L} and \texttt{R} are used to form a multiset rewriting rule, obtaining an MRS instantiation (\autoref{mrs_instantiation}).

\begin{lemma}\label{both_enabled}
Let $\mathtt{M}$ be a grounded multiset, $\mathtt{r} = \mathtt{lhs \Rightarrow rhs}$ a BCSL rule, and $\mu = (^{\bullet}\mu, \mu^{\bullet})$ its MRS instantiation. Then, $\mathtt{r}$ can be applied to $\mathtt{M}$ \emph{iff} $\mu$ can be applied to $\mathtt{M}$.
\end{lemma}

\begin{proof}
From construction of $\mathcal{M}$ we know that to BCSL rules correspond their MRS instantiations. 

\begin{itemize}
  \item[$\Rightarrow:$] if $\mathtt{r}$

  \begin{itemize}
    \item[(a)] \emph{can} be applied to $\mathtt{M}$, then there exists $\mathtt{L} \in \Theta(\mathtt{lhs})$ such that $\mathtt{L} \subseteq \mathtt{M}$ (follows from~\autoref{bcsl_rewriting}). That means there has to exist an MRS instantiation $\mu = (^{\bullet}\mu, \mu^{\bullet})$ in $\mathcal{X}$ of rule $\mathtt{r}$ such that it can be applied to $\mathtt{M}$ because $^{\bullet}\mu \equiv \mathtt{L}$ and therefore $^{\bullet}\mu \subseteq \mathtt{M}$ and $\mu$ is enabled.

    \item[(b)] \emph{can not} be applied to $\mathtt{M}$, then for all $\mathtt{L} \in \Theta(\mathtt{lhs})$ holds that $\mathtt{L} \not\subseteq \mathtt{M}$ (follows from~\autoref{bcsl_rewriting}). That means that any MRS instantiation $\mu = (^{\bullet}\mu, \mu^{\bullet})$ in $\mathcal{X}$ of rule $\mathtt{r}$ can not be applied to $\mathtt{M}$ because $^{\bullet}\mu \equiv \mathtt{L}$ and therefore $^{\bullet}\mu \not\subseteq \mathtt{M}$ and $\mu$ is not enabled.
  \end{itemize}

  \item[$\Leftarrow:$] Symmetrically, if $\mu = (^{\bullet}\mu, \mu^{\bullet})$

  \begin{itemize}
    \item[(a)] \emph{can} be applied to $\mathtt{M}$, then $\mu$ is enabled and therefore $^{\bullet}\mu \subseteq \mathtt{M}$. That means there has to exist a rule $\mathtt{r}$ such that $\mu$ is its MRS instantiation with $\mathtt{L} \in \Theta(\mathtt{lhs})$ where $^{\bullet}\mu \equiv \mathtt{L}$. Therefore, also $\mathtt{L} \subseteq \mathtt{M}$ and $\mathtt{r}$ can be applied to $\mathtt{M}$.

    \item[(b)] \emph{can not} be applied to $\mathtt{M}$, then $\mu$ is not enabled and therefore $^{\bullet}\mu \not\subseteq \mathtt{M}$. That means that any rule $\mathtt{r}$ such that $\mu$ is its MRS instantiation with $\mathtt{L} \in \Theta(\mathtt{lhs})$ where $^{\bullet}\mu \equiv \mathtt{L}$, holds that $\mathtt{L} \not\subseteq \mathtt{M}$ and $\mathtt{r}$ can not be applied to $\mathtt{M}$. \qed
  \end{itemize}
\end{itemize}

\end{proof}

\begin{lemma}\label{both_apply}
Let $\mathtt{M}$ be a grounded multiset, $\mathtt{r} = \mathtt{lhs \Rightarrow rhs}$ a BCSL rule, and $\mu = (^{\bullet}\mu, \mu^{\bullet})$ its MRS instantiation. Then, by applying $\mathtt{r}$ to $\mathtt{M}$, we get a set of possible multisets. Among them, there is a multiset $\mathtt{M'}$ which can be obtained by applying $\mu$ to $\mathtt{M}$.
\end{lemma}

\begin{proof}
Follows from the definition of BCSL rewriting (\autoref{bcsl_rewriting}) where instantiations of both $\mathtt{lhs}$ and $\mathtt{rhs}$ of the rule are created, which corresponds to the MRS instantiation $\mu$ (\autoref{mrs_instantiation}). Then, instantiated agents from $\mathtt{lhs}$ are subtracted from, and $\mathtt{rhs}$ agents are added to the current state, which is in parallel with the MRS approach. Finally, from~\autoref{states_equality} we know that states in BCSL directly correspond to multisets in MRS, which forms the same basis for both formalisms. \qed
\end{proof}

Having such constructed MRS $\mathcal{M}$, we need to show that its behaviour (set of runs) corresponds to the behaviour (transition system) of the BCSL model.

\begin{theorem}{~}
\label{mrs_is_bcsl}

For any BCSL model $\mathcal{B} = (\mathcal{R}, \Sigma_a, \Sigma_s, \mathtt{M_0})$ there exists an MRS $\mathcal{M} = (\mathcal{X}, \mathtt{M_0})$ with $\mathfrak{L}(\mathtt{LTS}_\varepsilon(\mathcal{B})) = \mathfrak{L}(\mathcal{M})$.
\end{theorem}

When we construct the MRS $\mathcal{M}$ using approach described in~\autoref{mrs_construction}, the proof of the theorem boils down to proving that for any grounded multiset $\mathtt{M}$ the following two implications hold:

\begin{itemize}
  \item[$\Rightarrow:$] for any BCSL rule $\mathtt{r} \in \mathcal{R}$ it holds that if $\mathcal{B}$ can apply \texttt{r} to \texttt{M} then there exists MRS rule $\mu \in \mathcal{X}$ such that $\mathcal{M}$ can apply $\mu$ to $\mathtt{M}$ 

  \item[$\Leftarrow:$] for any MRS rule $\mu \in \mathcal{X}$ it holds that if $\mathcal{M}$ can apply $\mu$ to \texttt{M} then there exists $\mathtt{r} \in \mathcal{R}$ such that $\mathcal{M}$ can apply \texttt{r} to $\mathtt{M}$ 
\end{itemize}

\noindent and in both cases we obtain the same multiset $\mathtt{M}'$.

\begin{proof}
Follows from construction of $\mathcal{M}$ (construction of corresponding MRS instantiations of rules), \autoref{both_enabled} (either both rules are enabled or neither of them is) and~\autoref{both_apply} (both rules create identical results). \qed
\end{proof}

\section{Regulations}

BCSL models manifest strong nondeterminism, which is natural but often not desired to some extent. Additional knowledge about the described biological system can further reduce possible model behaviour scenarios. These are usually introduced by defining quantitative properties~\cite{trojak2020parameter}. However, these properties are not always easy to define, and alternative mechanisms are needed.

In the following, we provide an introduction to \emph{regulation} approaches applied to BCSL. These were introduced in~\cite{trojak2021regulated} for MRS. To formally establish them in the context of BCSL, we assume the corresponding MRS is constructed first (\autoref{mrs_is_bcsl}) and the regulation is applied to it. This can be done because the constructed MRS shares rule labels with the original BCSL model and states and their content are of the same type (\autoref{states_equality}).

\subsection*{Regular rewriting}

In regular rewriting, there is given a $\omega$-regular language $\zeta$ over rules. This explicitly defines sequences of rules that can be used. Only runs with the rule sequence from this language are allowed. Typically, we define the language $\zeta$ using a regular expression.

For example, we define a regular expression $(\mathtt{r1\_S} . \mathtt{r1\_T} . \mathtt{r2} ~|~ \mathtt{r1\_T} . \mathtt{r1\_S})$ as regulation for model from~\autoref{running_model}. This RE makes sure that first both activation rules are used and then the molecule is exported out of the cell, depending on the order of activation. The effect of regulation on set of runs is depicted in~\autoref{running_lts_regular}.

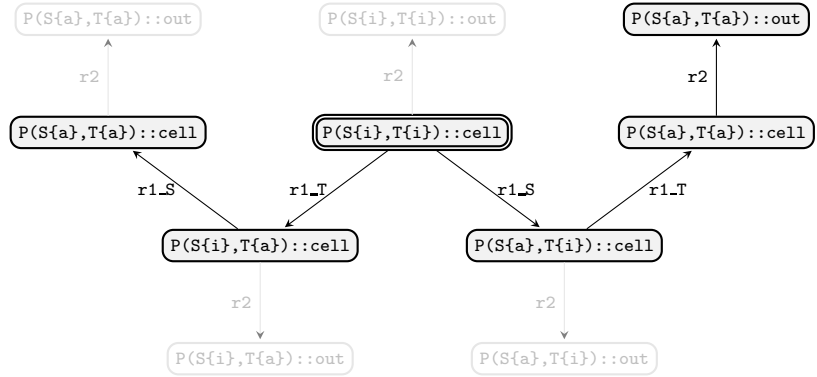
\begin{figure}

\begin{center}
\begin{tikzpicture}[shorten >=1pt, font=\scriptsize]
  \node[state,accepting] at (0,0) (q_0) {\lstinline[style=EBNF]!P(S{i},T{i})::cell!};
  \node[hidden] at (0,1.5) (q_4) {\lstinline[style=EBNF]!P(S{i},T{i})::out!};
  \node[state] at (-2,-1.5) (q_1) {\lstinline[style=EBNF]!P(S{i},T{a})::cell!};
  \node[state] at (2,-1.5) (q_2) {\lstinline[style=EBNF]!P(S{a},T{i})::cell!};
  \node[state] at (-4,0) (q_3_1) {\lstinline[style=EBNF]!P(S{a},T{a})::cell!};
  \node[state] at (4,0) (q_3_2) {\lstinline[style=EBNF]!P(S{a},T{a})::cell!};
  \node[hidden] at (-2,-3) (q_5) {\lstinline[style=EBNF]!P(S{i},T{a})::out!};
  \node[hidden] at (2,-3) (q_6) {\lstinline[style=EBNF]!P(S{a},T{i})::out!};
  \node[hidden] at (-4,1.5) (q_7_1) {\lstinline[style=EBNF]!P(S{a},T{a})::out!};
  \node[state] at (4,1.5) (q_7_2) {\lstinline[style=EBNF]!P(S{a},T{a})::out!};

  \path[->] 
            (q_0) edge node [left] {$\mathtt{r1\_T}$} (q_1)
            (q_0) edge node [right] {$\mathtt{r1\_S}$} (q_2)
            (q_1) edge node [left] {$\mathtt{r1\_S}$} (q_3_1)
            (q_2) edge node [right] {$\mathtt{r1\_T}$} (q_3_2)
            (q_3_2) edge node [left] {$\mathtt{r2}$} (q_7_2);
  \path[every edge/.style={gray,draw=gray!20,text=gray!50}]
            (q_0) edge node [left] {$\mathtt{r2}$} (q_4)
            (q_1) edge node [left] {$\mathtt{r2}$} (q_5)
            (q_2) edge node [left] {$\mathtt{r2}$} (q_6)
            (q_3_1) edge node [left] {$\mathtt{r2}$} (q_7_1);
\end{tikzpicture}
\end{center}

\caption{\textbf{The set of runs of the model with regular regulation.} The runs correspond to the model from~\autoref{running_model} with applied regular regulation. The runs are represented as possible paths in the tree-like graph, starting in the double circled state (the initial state). The grey states and transition are absent due to effects of the regulation.}\label{running_lts_regular}
\end{figure}

\subsection*{Ordered rewriting}

Ordered regulation defines a partial order on rules. Then it is not allowed to apply a rule immediately after the rule which is higher in the order. Runs that violate this property are not allowed.

For example, we define a partial order $\mathtt{r1\_S} < \mathtt{r2}, \mathtt{r1\_T} < \mathtt{r2}$ as regulation for model from Figure~\ref{running_model}. This order makes sure that rule $\mathtt{r2}$ is never used after rule $\mathtt{r1\_S}$ neither rule $\mathtt{r1\_T}$. The effect of regulation on set of runs is depicted in~\autoref{running_lts_ordered}.

\begin{figure}

\begin{center}
\begin{tikzpicture}[shorten >=1pt, font=\scriptsize]
  \node[state,accepting] at (0,0) (q_0) {\lstinline[style=EBNF]!P(S{i},T{i})::cell!};
  \node[state] at (0,1.5) (q_4) {\lstinline[style=EBNF]!P(S{i},T{i})::out!};
  \node[state] at (-2,-1.5) (q_1) {\lstinline[style=EBNF]!P(S{i},T{a})::cell!};
  \node[state] at (2,-1.5) (q_2) {\lstinline[style=EBNF]!P(S{a},T{i})::cell!};
  \node[state] at (-4,0) (q_3_1) {\lstinline[style=EBNF]!P(S{a},T{a})::cell!};
  \node[state] at (4,0) (q_3_2) {\lstinline[style=EBNF]!P(S{a},T{a})::cell!};
  \node[hidden] at (-2,-3) (q_5) {\lstinline[style=EBNF]!P(S{i},T{a})::out!};
  \node[hidden] at (2,-3) (q_6) {\lstinline[style=EBNF]!P(S{a},T{i})::out!};
  \node[hidden] at (-4,1.5) (q_7_1) {\lstinline[style=EBNF]!P(S{a},T{a})::out!};
  \node[hidden] at (4,1.5) (q_7_2) {\lstinline[style=EBNF]!P(S{a},T{a})::out!};

  \path[->] (q_0) edge node [left] {$\mathtt{r2}$} (q_4)
            (q_0) edge node [left] {$\mathtt{r1\_T}$} (q_1)
            (q_0) edge node [right] {$\mathtt{r1\_S}$} (q_2)
            (q_1) edge node [left] {$\mathtt{r1\_S}$} (q_3_1)
            (q_2) edge node [right] {$\mathtt{r1\_T}$} (q_3_2);
  \path[every edge/.style={gray,draw=gray!20,text=gray!50}]
            (q_3_2) edge node [left] {$\mathtt{r2}$} (q_7_2)
            (q_1) edge node [left] {$\mathtt{r2}$} (q_5)
            (q_2) edge node [left] {$\mathtt{r2}$} (q_6)
            (q_3_1) edge node [left] {$\mathtt{r2}$} (q_7_1);
\end{tikzpicture}
\end{center}

\caption{\textbf{The set of runs of the model with ordered regulation.} The runs correspond to the model from~\autoref{running_model} with applied ordered regulation. The runs are represented as possible paths in the tree-like graph, starting in the double circled state (the initial state). The grey states and transition are absent due to effects of the regulation.}\label{running_lts_ordered}
\end{figure}
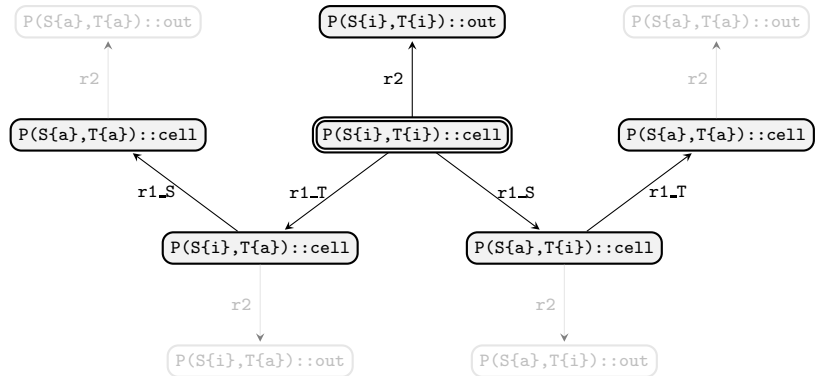

\subsection*{Programmed rewriting}

Programmed regulation defines a set of successor rules to every rule. When a particular rule is used, only its successors are allowed to be used next. Similarly to the previous regulation, Runs which violate this property are not allowed.

For example, we use successor function defined as $\zeta(\mathtt{r1\_S}) = \{\mathtt{r2}, \mathtt{r1\_T}\}$, $\zeta(\mathtt{r1\_T}) = \{ \mathtt{r1\_S} \}$, and $\zeta(\mathtt{r2}) = \emptyset$ as regulation for model from Figure~\ref{running_model}. This function makes sure that rule $\mathtt{r2}$ is used only after rule $\mathtt{r1\_S}$, never after rule $\mathtt{r1\_T}$. The effect of regulation on set of runs is depicted in~\autoref{running_lts_programmed}.

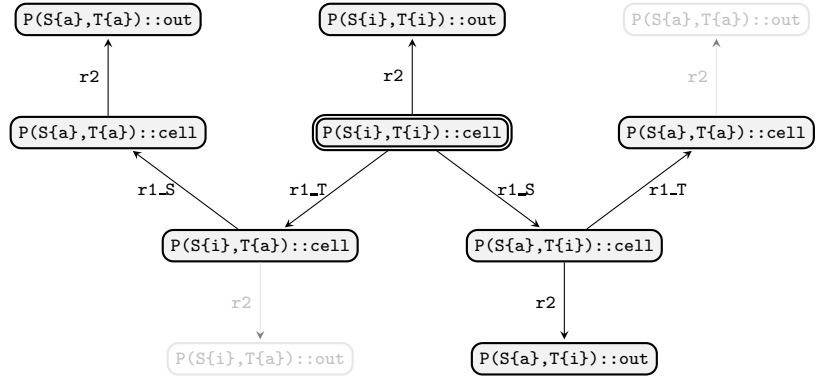
\begin{figure}

\begin{center}
\begin{tikzpicture}[shorten >=1pt, font=\scriptsize]
  \node[state,accepting] at (0,0) (q_0) {\lstinline[style=EBNF]!P(S{i},T{i})::cell!};
  \node[state] at (0,1.5) (q_4) {\lstinline[style=EBNF]!P(S{i},T{i})::out!};
  \node[state] at (-2,-1.5) (q_1) {\lstinline[style=EBNF]!P(S{i},T{a})::cell!};
  \node[state] at (2,-1.5) (q_2) {\lstinline[style=EBNF]!P(S{a},T{i})::cell!};
  \node[state] at (-4,0) (q_3_1) {\lstinline[style=EBNF]!P(S{a},T{a})::cell!};
  \node[state] at (4,0) (q_3_2) {\lstinline[style=EBNF]!P(S{a},T{a})::cell!};
  \node[hidden] at (-2,-3) (q_5) {\lstinline[style=EBNF]!P(S{i},T{a})::out!};
  \node[state] at (2,-3) (q_6) {\lstinline[style=EBNF]!P(S{a},T{i})::out!};
  \node[state] at (-4,1.5) (q_7_1) {\lstinline[style=EBNF]!P(S{a},T{a})::out!};
  \node[hidden] at (4,1.5) (q_7_2) {\lstinline[style=EBNF]!P(S{a},T{a})::out!};

  \path[->] (q_0) edge node [left] {$\mathtt{r2}$} (q_4)
            (q_0) edge node [left] {$\mathtt{r1\_T}$} (q_1)
            (q_0) edge node [right] {$\mathtt{r1\_S}$} (q_2)
            (q_1) edge node [left] {$\mathtt{r1\_S}$} (q_3_1)
            (q_2) edge node [right] {$\mathtt{r1\_T}$} (q_3_2)
            (q_3_1) edge node [left] {$\mathtt{r2}$} (q_7_1)
            (q_2) edge node [left] {$\mathtt{r2}$} (q_6);
  \path[every edge/.style={gray,draw=gray!20,text=gray!50}]
            (q_1) edge node [left] {$\mathtt{r2}$} (q_5)
            (q_3_2) edge node [left] {$\mathtt{r2}$} (q_7_2);
\end{tikzpicture}
\end{center}

\caption{\textbf{The set of runs of the model with programmed regulation.} The runs correspond to the model from~\autoref{running_model} with applied programmed regulation. The runs are represented as possible paths in the tree-like graph, starting in the double circled state (the initial state). The grey states and transition are absent due to effects of the regulation.}\label{running_lts_programmed}
\end{figure}

\subsection*{Conditional rewriting}

Conditional rewriting defines a \emph{prohibited} context to each rule, that is, a multiset of \emph{grounded} agents which cannot be present in the current state. Conditional regulation is based on local information and does not need any history of applied rules. Runs that violate this property are not allowed.

Please note that without loss of generality, the prohibited context can contain a set of prohibited multisets, and for then each of them it has to hold that it is not a subset of the current state.

For example, we define prohibited context $\zeta(\mathtt{r2}) = \{$ \lstinline[style=BCSL]!P(S{a},T{i})::cell! $\}$ as regulation for model from Figure~\ref{running_model}. This $\zeta$ makes sure that rule $\mathtt{r2}$ is never used when agent \lstinline[style=BCSL]!P(S{a},T{i})::cell! is present in the current state. The effect of regulation on set of runs is depicted in~\autoref{running_lts_conditional}.

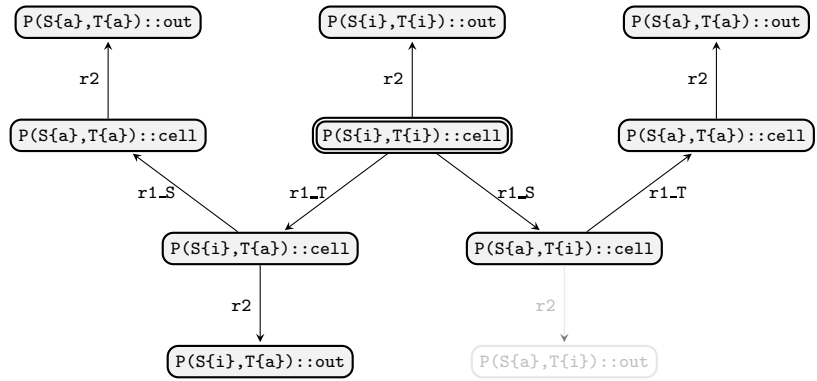
\begin{figure}

\begin{center}
\begin{tikzpicture}[shorten >=1pt, font=\scriptsize]
  \node[state,accepting] at (0,0) (q_0) {\lstinline[style=EBNF]!P(S{i},T{i})::cell!};
  \node[state] at (0,1.5) (q_4) {\lstinline[style=EBNF]!P(S{i},T{i})::out!};
  \node[state] at (-2,-1.5) (q_1) {\lstinline[style=EBNF]!P(S{i},T{a})::cell!};
  \node[state] at (2,-1.5) (q_2) {\lstinline[style=EBNF]!P(S{a},T{i})::cell!};
  \node[state] at (-4,0) (q_3_1) {\lstinline[style=EBNF]!P(S{a},T{a})::cell!};
  \node[state] at (4,0) (q_3_2) {\lstinline[style=EBNF]!P(S{a},T{a})::cell!};
  \node[state] at (-2,-3) (q_5) {\lstinline[style=EBNF]!P(S{i},T{a})::out!};
  \node[hidden] at (2,-3) (q_6) {\lstinline[style=EBNF]!P(S{a},T{i})::out!};
  \node[state] at (-4,1.5) (q_7_1) {\lstinline[style=EBNF]!P(S{a},T{a})::out!};
  \node[state] at (4,1.5) (q_7_2) {\lstinline[style=EBNF]!P(S{a},T{a})::out!};

  \path[->] (q_0) edge node [left] {$\mathtt{r2}$} (q_4)
            (q_0) edge node [left] {$\mathtt{r1\_T}$} (q_1)
            (q_0) edge node [right] {$\mathtt{r1\_S}$} (q_2)
            (q_1) edge node [left] {$\mathtt{r2}$} (q_5)
            (q_1) edge node [left] {$\mathtt{r1\_S}$} (q_3_1)
            (q_2) edge node [right] {$\mathtt{r1\_T}$} (q_3_2)
            (q_3_1) edge node [left] {$\mathtt{r2}$} (q_7_1)
            (q_3_2) edge node [left] {$\mathtt{r2}$} (q_7_2);
  \path[every edge/.style={gray,draw=gray!20,text=gray!50}]
            (q_2) edge node [left] {$\mathtt{r2}$} (q_6);
\end{tikzpicture}
\end{center}

\caption{\textbf{The set of runs of the model with conditional regulation.} The runs correspond to the model from~\autoref{running_model} with applied conditional regulation. The runs are represented as possible paths in the tree-like graph, starting in the double circled state (the initial state). The grey states and transition are absent due to effects of the regulation.}\label{running_lts_conditional}
\end{figure}

\subsection*{Concurrent-free rewriting}

Concurrent rules are those which consume common agents. Concurrent-free rewriting assigns a priority to one of the concurrent rules. Whenever multiple concurrent rules are applicable in a state, only the prioritised one can be used. Runs that violate this property are not allowed.

For example, we define prioritisation $\zeta = \{ (\mathtt{r1_S}, \mathtt{r2}), (\mathtt{r1_T}, \mathtt{r2}) \}$ as regulation for model from Figure~\ref{running_model}. This $\zeta$ makes sure that rules $\mathtt{r1_S}$ and $\mathtt{r1_T}$ have always priority over rule $\mathtt{r2}$. The effect of regulation on set of runs is depicted in~\autoref{running_lts_concurrent_free}.

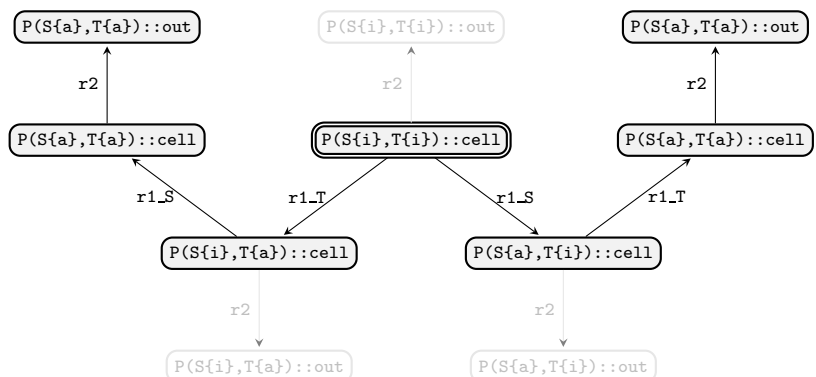
\begin{figure}

\begin{center}
\begin{tikzpicture}[shorten >=1pt, font=\scriptsize]
  \node[state,accepting] at (0,0) (q_0) {\lstinline[style=EBNF]!P(S{i},T{i})::cell!};
  \node[hidden] at (0,1.5) (q_4) {\lstinline[style=EBNF]!P(S{i},T{i})::out!};
  \node[state] at (-2,-1.5) (q_1) {\lstinline[style=EBNF]!P(S{i},T{a})::cell!};
  \node[state] at (2,-1.5) (q_2) {\lstinline[style=EBNF]!P(S{a},T{i})::cell!};
  \node[state] at (-4,0) (q_3_1) {\lstinline[style=EBNF]!P(S{a},T{a})::cell!};
  \node[state] at (4,0) (q_3_2) {\lstinline[style=EBNF]!P(S{a},T{a})::cell!};
  \node[hidden] at (-2,-3) (q_5) {\lstinline[style=EBNF]!P(S{i},T{a})::out!};
  \node[hidden] at (2,-3) (q_6) {\lstinline[style=EBNF]!P(S{a},T{i})::out!};
  \node[state] at (-4,1.5) (q_7_1) {\lstinline[style=EBNF]!P(S{a},T{a})::out!};
  \node[state] at (4,1.5) (q_7_2) {\lstinline[style=EBNF]!P(S{a},T{a})::out!};

  \path[->] 
            (q_0) edge node [left] {$\mathtt{r1\_T}$} (q_1)
            (q_0) edge node [right] {$\mathtt{r1\_S}$} (q_2)
            (q_1) edge node [left] {$\mathtt{r1\_S}$} (q_3_1)
            (q_2) edge node [right] {$\mathtt{r1\_T}$} (q_3_2)
            (q_3_2) edge node [left] {$\mathtt{r2}$} (q_7_2)
            (q_3_1) edge node [left] {$\mathtt{r2}$} (q_7_1);
  \path[every edge/.style={gray,draw=gray!20,text=gray!50}]
            (q_0) edge node [left] {$\mathtt{r2}$} (q_4)
            (q_1) edge node [left] {$\mathtt{r2}$} (q_5)
            (q_2) edge node [left] {$\mathtt{r2}$} (q_6);
\end{tikzpicture}
\end{center}

\caption{\textbf{The set of runs of the model with concurrent-free regulation.} The runs correspond to the model from~\autoref{running_model} with applied concurrent-free regulation. The runs are represented as possible paths in the tree-like graph, starting in the double circled state (the initial state). The grey states and transition are absent due to effects of the regulation.}\label{running_lts_concurrent_free}
\end{figure}

\section{Summary}

In this short paper, we first introduce multiset rewriting systems MRS~\cite{trojak2021regulated} and BioChemical Space language~\cite{TROJAK202091}. Then, in~\autoref{mrs_construction} we show how for any BCSL model, we can construct an MRS such that the corresponding set of runs are equal for both systems. Finally, we introduce \emph{regulations} in the context of BCSL, formally influencing the runs of respective MRS. This way, we can use regulated BCSL models while they hold properties shown in~\cite{trojak2021regulated}.

\bibliographystyle{plainurl}
\bibliography{papers}

\end{document}